\documentclass{article}

\usepackage{url}
\usepackage{amsthm}
\usepackage{amsfonts}
\usepackage[intlimits]{amsmath}
\usepackage{graphicx}
\usepackage{authblk}

\newtheorem{lemma}{Lemma}

\newtheorem{definition}{Definition}

\newtheorem{observation}{Observation}
\newtheorem{conjecture}[lemma]{Conjecture}
\newtheorem{theorem}[lemma]{Theorem}

\begin{document}
\title{Proving the Pressing Game Conjecture on Linear Graphs
  \thanks{This paper presents the results of the undergraduate research
  of E. Bixby and T. Flint in the 2012 Fall semester at the Budapest
  Semesters in Mathematics}}

\author[1,2]{Eliot Bixby}
\author[1,3]{Toby Flint}
\author[1,4]{Istv\'an Mikl\'os}

\affil[1]{Budapest Semesters in Mathematics \\H-1071 Budapest,
  Bethlen G\'abor t\'er 2 \\ Hungary}
\affil[2]{Oberlin College \\ 37 W College St, Oberlin, OH 44074 \\ USA}
\affil[3]{Hillsdale College \\ 33 East College Street, Hillsdale, MI
  49242 \\ USA}
\affil[4]{R\'enyi Institute \\ H-1053 Budapest, Re\'altanoda utca
  13-15 \\ Hungary}

\maketitle

\begin{abstract}
The pressing game on black-and-white graphs is the following: Given a
graph $G(V,E)$ with its vertices colored with black and white, any
black vertex $v$ can be pressed, which has the following effect: (a)
all neighbors of $v$ change color, i.e. white neighbors become black and
\emph{vice versa}, (b) all pairs of neighbors of $v$ change
connectivity, i.e. connected pairs become unconnected, unconnected ones
become connected, (c) and finally, $v$ becomes a separated white
vertex. The aim of the game is to transform $G$ into an all white,
empty graph. It is a known result that the all white empty graph is
reachable in the pressing game if each component of $G$ contains at
least one black vertex, and for a fixed graph, any successful
transformation has the same number of pressed vertices. 

The pressing game conjecture is that any successful pressing path can
be transformed into any other successful pressing path with small
alterations. Here we prove the conjecture for linear graphs. The
connection to genome rearrangement and sorting signed permutations
with reversals is also discussed. 
\end{abstract}


\section{Introduction}
Sorting signed permutations by reversals (or inversions as biologists
call it) is the first genome rearrangement model introduced in the
scientific literature. The hypothesis that reversals change the order
and orientation of genes -- so-called genetic factors in that time --
arose in a paper published in 1921 \cite{sturtevant21} and implicitly
was verified by microscopic inferring of chromosomes a couple of
decades later \cite{sn41}. In the same time, geneticists realized that
``The mathematical properties of series of letters subjected to the
operation of successive inversions do not appear to be worked out''
\cite{st37}. This computational problem has been rediscovered at the
end of the XX. century, and the solution to it, now called as the
Hannenhalli-Pevzner theorem has been published in 1995 and 1999
\cite{hp99}. 

The Hannenhalli-Pevzner theorem gives a polynomial running time
algorithm to find one scenario with the minimum number of reversals
necessary to sort a signed permutation. However, there might be
multiple solutions, and the number of solutions typically grows 
exponentially with the length of the permutation. Therefore, a(n
almost) uniform sampler is required which gives a set of solutions
from which statistical properties of the solutions can be
calculated. A typical approach for sampling is the Markov chain Monte
Carlo method. It starts with an arbitrary solution, and applies random
perturbations on it thus exploring the solution space. In case of most
parsimonious reversal sorting scenarios, two approaches are considered as
perturbing the current solution:
\begin{itemize}
\item The first approach encodes the most parsimonious reversal
  sorting scenarios with the intermediate permutations visited. Then
  it cuts out a random window from this path, and gives a random new
  sorting scenario between the permutations at the beginning and end of
  the window
\item The second approach encodes the scenarios with the series of
  mutations applied, and perturbs them in a sophisticated way,
  described in details later in this paper.
\end{itemize}
A Markov chain for sampling purposes should fulfill two conditions:
{\bf (a)} it must converge to the uniform distributions of all
possibilities, and a necessary condition for it that it must be
irreducible, namely, from any solution the chain must be able to get
to any another solution and {\bf (b)} the convergence must be fast.

The problem with the first approach mentioned above is that it is
provenly slowly mixing \cite{mms2010}. This means that the necessary
number of steps in the Markov chain to get sufficiently close to the
uniform distribution grows exponentially with the length of the size
of the permutation. Therefore this approach is not applicable in
practice. 

The problem with the second approach is that we even do not know if it
is irreducible, nor that it is rapidly mixing. In this paper, we want
to take a step towards proving that it is an irreducible Markov chain.

The paper is organized in the following way. In Section~\ref{sec:pre},
we define the problem of sorting by reversals, and the combinatorics
tools necessary: the graph of desire and reality and the overlap
graph. Then we introduce the pressing game on the black-and-white
graphs, and show that they correspond to the shortest reversal
scenarios in case of a biologically important subset of
permutations. We finish the section for stating the pressing game
conjecture. If this was proven then this would give a proof for the
irreducibility of the Markov chain applying the above mentioned second
approach. In Section~\ref{sec:proof},  we prove the conjecture for
linear graphs. The paper is finished with a discussion and
conclusions. 

\section{Preliminaries}\label{sec:pre}

\begin{definition}
A \emph{signed permutation} is a permutation of numbers from $1$ to
$n$, where each number has a $+$or $-$ sign.
\end{definition}
While the number of $n$ long permutations is $n!$, the number of $n$
long signed permutations is $2^n\times n!$.
\begin{definition}
A \emph{reversal} takes any consecutive part of a signed permutation
and change both the order of the numbers and the sign of each
number. It is also allowed that a reversal takes only one single
number from the signed permutations, in that case, it changes the sign
of this number.
\end{definition}
For example, the following reversal flips the $-3\ +6\ -5\ +4\ +7$
segment:

$$+8\ -1\ -3\ +6\ -5\ +4\ +7\ -9\ +2 \Rightarrow +8\ -1\ -7\ -4\ +5\
-6\ +3 -9\ +2$$

The sorting by reversals problem asks for the minimum number of
reversals necessary to transform a signed permutation into the
identity permutation, ie. the signed permutation $+1\ +2\ \ldots\
+n$. This number is called the reversal distance, and the reversal
distance of a signed permutation $\pi$ is denoted by $d_{REV}(\pi)$.
To solve this problem, we have to introduce two discrete
mathematical objects, the graph of desire and reality and the overlap
graph. The graph of desire and reality is a drawn graph, ie. not only
the topology (which vertices are connected) but also its drawing
matter. The overlap graph is a graph in terms of standard graph theory.

The graph of desire and reality of a signed permutation can be
constructed in the following way. Each signed number is replaced with
two unsigned number, $+i$ becomes $2i-1, 2i$, $-i$ becomes $2i,
2i-1$. The so-obtained $2n$ long permutation is framed between $0$ and
$2n+1$. Each number including $0$ and $2n+1$ will represent one vertex
in the graph of desire and reality. They are drawn in the same order
as they appear in the permutation, so the graph of desire and reality
is not only a graph, its drawing is also important.

We index the positions of the vertices starting with $1$, and each
pair of vertices in positions $2i-1$ and $2i$ are connected with an
edge. We call these edges the reality edges. Each pair of vertices for
numbers $2i$ and $2i+1$, $i = 0, 1,\ldots n$ are  connected with an
arc, and they are named the desire edges. The explanation for these
names is that the reality edges describes what we see in the current
permutation, and the desire edges tell what neighbors we would like to
see to get the $+1, +2, \ldots +n$ permutation: we would like that $1$
be next to $0$, $3$ be next to $2$, etc.

The overlap graph is constructed from the graph of desire and reality
in the following way. The vertices of the overlap graph are the desire
edges in the graph of desire and reality. The vertices are colored, a
vertex in the overlap graph is black if the number of vertices below
the desire edge it represents in the graph of desire and reality is
odd. A vertex is white if the number of vertices is even. Two vertices
are connected if the intervals that the corresponding desire edges
span overlap but neither contain the other. On
Figure~\ref{fig:example}, we give an example for the graph of desire
and reality and overlap graph.
\begin{figure}
\center{
\includegraphics[angle=0, width=2.9in]{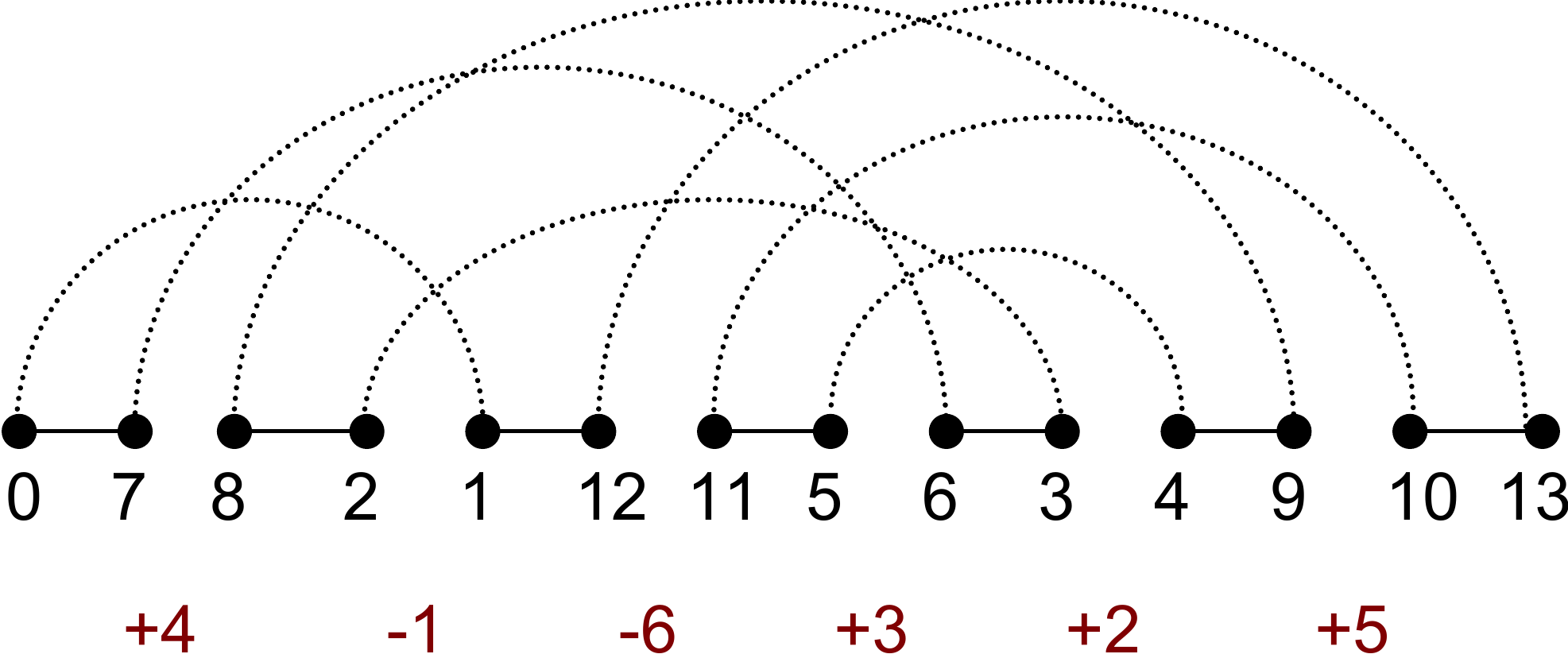}
\includegraphics[angle=0, width=2in]{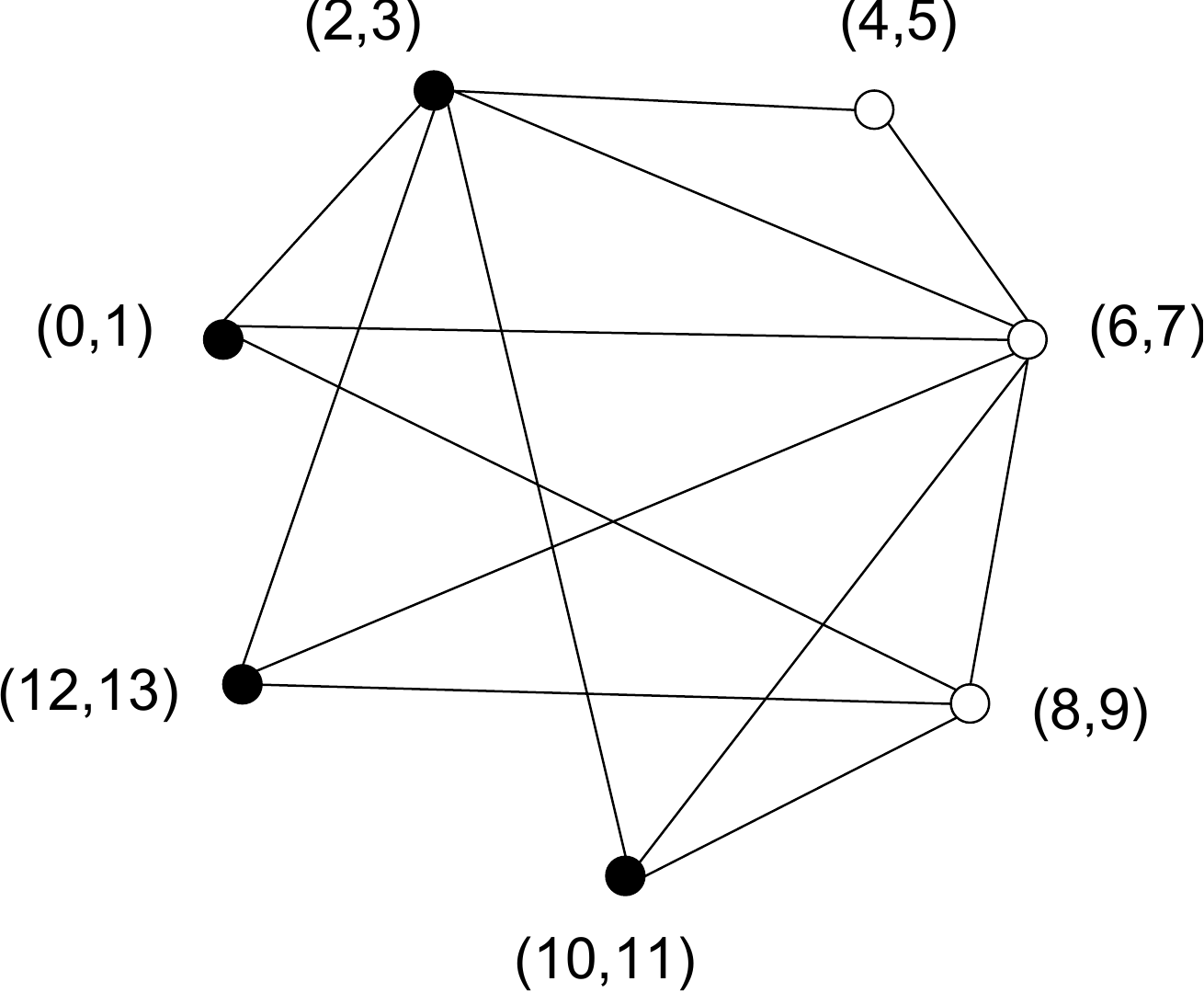}}
\caption{The graph of desire and reality and the overlap graph of the
  signed permutation $+4\ -1\ -6\ +3\ +2\ +5$}\label{fig:example}
\end{figure}

The overlap graph might fall into components. A vertex, as well as its
corresponding desire edge is called oriented if the vertex is black,
namely, the corresponding desire edge spans odd number of vertices. A
vertex and its corresponding  desire edge is unoriented if the vertex is
white. A component is called oriented if it contains at least one
black vertex, if the component contains only white vertices, it is
called unoriented. A component is non-trivial if it contains more than
one vertex. Some of the non-trivial unoriented components are
hurdles. We skip the precise definition of hurdles here as we do not
need it. A permutation is called fortress, if the number of its hurdles
is odd, with some prescribed properties, also not detailed here.

Any reversal changes the topology of the graph of desire and reality
on two reality vertices. Any desire edge is a neighbor of two reality
edges, and we say that the reversal acts on this desire edge if it
changes the topology on the two neighbor reality edges.

How do such reversals change the graph of desire and reality and thus
the overlap graph? We set up a Lemma below explaining this. 

\begin{lemma}\label{lem:press}
Let $v$ be an orientd desire edge on which the reversal acts. Then the
reversal
\begin{itemize}
\item change the orientation of any desire edge crossing $v$
\item change the overlap of any pair of desire edges crossing $v$
\item the desire edge itself become an unoriented edge without any
  overlap with any other edges.
\end{itemize}
\end{lemma}
\begin{proof}
The reversal flips one of the 'legs' of each overlapping desire edge,
namely, the reality edge connected to the desire edge. Therefore it
changes the parity of the vertices below the desire edge and thus the
orientation of it.

Two edges which are both overlap with $v$ but not with each other, can
overlap only from the two ends of $v$, see also Fig.~\ref{fig:click},
case {\bf I}. A reversal acting on $v$ will
flip one-one of their endpoints, so they will indeed overlap. If two
edges overlap with $v$, but by definition not with each other since
the interval of one of them contains the interval of the other, then
they come from one end of $v$. It is easy to see that after the
reversal they will overlap by definition, see Fig.~\ref{fig:click},
case {\bf II}. It is also easy to see that
overlapping pair of edges which are also overlap with each other are
the cases on the right hand side of Fig.~\ref{fig:click}, so after the
reversal, they will not overlap.

Finally, the oriented edge on which the reversal acts becomes an
unoriented edge forming a small cycle with a reality edge, and thus it
cannot overlap with any other desire edge.

\begin{figure}
\center{
\includegraphics[angle=0, width=3in]{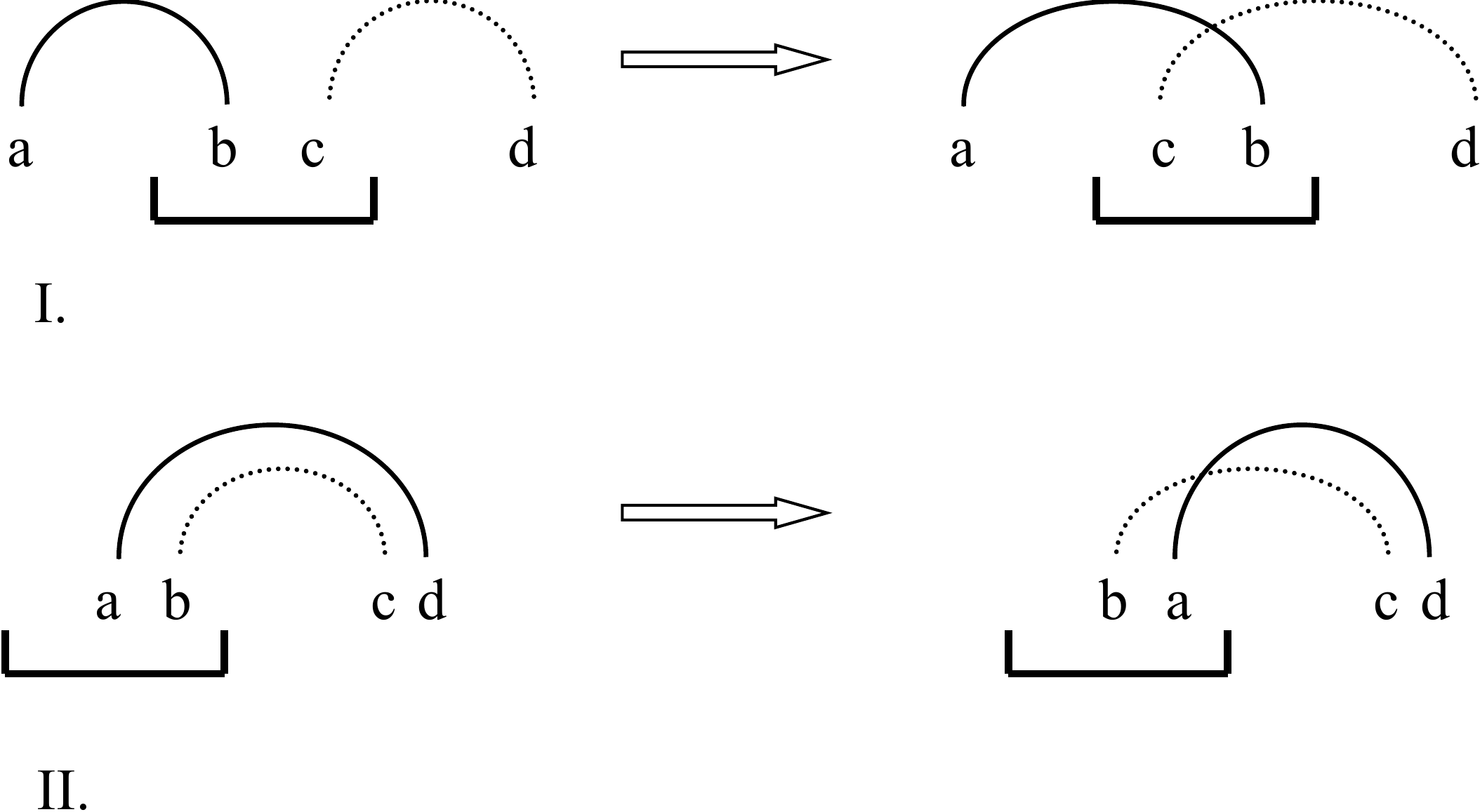}}
\caption{This picture show how a reversal can change the overlap of
  two desire edges. The reverted fragment is indicated with a thick
  black line.}\label{fig:click}
\end{figure}
\end{proof}

This lemma also shows the connection between sorting by reversals and
the pressing game on black and white graphs: a pressing of a black
vertex is equivalent with a reversal acting on the corresponding
desire edge. Below we define the pressing game on black-and-white
graphs:
\begin{definition}
Given a graph $G(V,E)$ with its vertices colored with black and
white. Any black vertex $v$ can be pressed, which has the following
effect: (a) all neighbors of $v$ change color, white neighbors become
black and \emph{vice versa}, (b) all pair of neighbors of $v$ change
connectivity, connected pairs become unconnected, unconnected ones
become connected, (c) and finally, $v$ becomes a separated white
vertex. The aim of the game is to transform $G$ into an all white,
empty graph. 
\end{definition}

If each component of $G$ contains at least one black vertex, then the
pressing game always has at least one solution, as it turns out from
the Hannenhalli-Pevzner theorem. 

\begin{theorem}
{\bf (Hannenhalli-Pevzner), \cite{hp99}}
$$d_{REV}(\pi) = n+1-c(\pi)+h(\pi)+f(\pi)$$
where $n$ is the length of the permutation $\pi$, $c(\pi)$ is the
number of cycles in the graph of desire and reality, $h(\pi)$ is the
number of hurdles in the permutation and $f(\pi)$ is the fortress
indicator, it is $1$ if the permutation is a fortress, otherwise $0$.
\end{theorem}

It is easy to see that any reversal can increase the number of cycles
in the graph of desire and reality at most by $1$, hence the
Hannenhalli-Pevzner theorem also says if a permutation does not
contain any hurdle (and thus it is not a fortress) then any optimal
reversal sorting path increases the number of cycles to $n+1$ without
creating any hurdle. Below we state this theorem.
\begin{theorem}\label{theo:safe-reversal}
Let $\pi$ be a permutation which is not the identical permutation and
whose overlap graph does not contain any non-trivial unoriented
component. Then a reversal exists that acts on an oriented desire
edge, thus increases $c(\pi)$ by $1$ and does not create any
non-trivial unoriented component. 

Furthermore, if $G$ is an arbitrary  black-and-white graph such that
each component contains at least one black vertex, then at least one
black vertex can be pressed without making a non-trivial unoriented
component. 
\end{theorem}

The proof can be found in \cite{bergeron01}, and we skip it here. The
proof consider only the overlap graph, and in fact, it indeed works for every
black-and-white graph. A clear consequence is the following theorem.

\begin{theorem} Let $G$ be a black-and-white graph such that each
component on it contains at least one black vertex. Then $G$ can be
transformed into the all-white empty graph in the pressing game.
\end{theorem}
\begin{proof}
It is sufficient to use iteratively
Theorem~\ref{theo:safe-reversal}. Indeed, according to
Theorem~\ref{theo:safe-reversal}, we can find a black vertex $v$, such
that pressing it does not create a non-trivial all-white component, on
the other hand, $v$ become a separated white vertex, and it will
remain a separated white vertex afterward. Hence, the number of
vertices in non-trivial components decreases at least by one, and in a
finite number of steps, $G$ is transformed into the all-white, empty graph.
\end{proof}

Consider the set of vertices as an alphabet, any sequence over this
alphabet is called a pressing path. It is a valid pressing path when
each vertex is black when it is pressed, and it is successful, if it
is valid and leads to the all-white, empty graph. The length of the
pressing path is the number of vertices pressed in it. The following
theorem is also true. 

\begin{theorem}
Let $G$ be a black-and-white graph such that each
component on it contains at least one black vertex. Then each
successful pressing path of $G$ has the same length.
\end{theorem}

The proof can be found in \cite{hv2006}. We are ready to state the
pressing path conjecture.

\begin{conjecture}
Let $G$ be a black-and-white graph such that each
component on it contains at least one black vertex. Construct a
metagraph, $M$ whose vertices are the successful pressing paths on
$G$. Connect two vertices if the length of the longest common
subsequence of the pressing paths they represent is at most $4$
less than the common length of the pressing paths. The conjecture is
that $M$ is connected.
\end{conjecture}

The conjecture means that with small alterations, we can transform any
pressing path into any other pressing path, whatever $G$ is. The small
alteration means that we remove at most $4$, not necessary consecutive
vertices from a pressing path, and add at most $4$ vertices, not
necessarily to the same places where the old vertices were removed
from, and not necessarily to consecutive places. Although it is
generally not true that only the pressing paths are the reversal
sorting paths of a signed permutation, as there might be
cycle-increasing reversals not acting on a desire edge, for a class of
permutations, it is true. Specially, if a signed permutation is such
that in its graph of desire and reality each cycle contains only one
or two desire edges, then all cycle-increasing reversals act on desire
edges. These signed permutations are the permutations that can be
considered in the so-called infinite site model \cite{maetal2008}. 

In this paper, we prove the pressing game conjecture for linear
graphs. Actually, we can prove more, the metagraph will be already
connected if we require that neighbor vertices have longest common
subsequence at most $2$ less than the common length of their pressing paths.

\section{Proof of the Conjecture on Linear Graphs}\label{sec:proof}

The proof of our main theorem is recursive, and for this, we need the
following notations. Let $G$ be a black-and-white graph, and $v$ a
black vertex in it. Then $Gv$ denotes the graph we get by pressing
vertex $v$. Similarly, if $P$ is a valid pressing path of $G$ (namely, each
vertex is black when we want to press it, but $P$ does not necessary
yield the all-white, empty graph), then $GP$ denotes the graph we get
after pressing all vertices in $P$ in the indicated order. Finally,
let $P^k$ denote the suffix of $P$ starting in position $k+1$.

The simplicity of the linear graphs is that they have a simple
structure and furthermore, the pressing game on linear graphs is
self-reducible as the following observation states.
\begin{observation}\label{obs:self-red}
Let $G$ be a linear black-and-white graph and $v$ a black vertex in
it. Then $Gv$ is also a linear graph and the separated white vertex $v$.
\end{observation}
Since any separated white vertex does not have to be pressed again,
it is sufficient to consider $Gv \setminus \{v\}$, which is a linear
graph. We are ready to state and prove our main theorem.

\begin{theorem}\label{theo:main}
Let $G$ be an arbitrary, finite, linear black-and-white graph, and let $M$ be the
following graph. The vertices of $M$ are the successful pressing
paths on $G$, and two vertices are connected if the length of the
longest common subsequence of the pressing paths they represent is at
most $2$ less than the common length of the pressing paths. Then $M$
is connected.
\end{theorem}
\begin{proof}
It is sufficient to show that for any successful pressing paths $X$
and $Y = v_1 v_2 \ldots v_k$ there is a series $X_1, X_2, \ldots
X_m$ such that for any $i =1, 2, \ldots m-1$, the length of the
longest common subsequence of $X_i$ and $X_{i+1}$ is at most $2$ less
than the common length of the paths, and $X_m$ starts with
$v_1$. Indeed, then both $X_m$ and $Y$ starts with $v_1$, and both
$X_m^1$ and $Y^1$ are successful pressing paths on $G v_1 \setminus
\{v_1\}$. We can use the induction to transform $X_m$ into a pressing
path which starts $v_2$, then we consider its suffix which is a
successful pressing path on $G v_1 v_2 \setminus \{v_1,v_2\}$, etc. 

Furthermore, to show that $v_1$ can be moved to the first position to
the current pressing path, it is sufficient to show that it can be
moved towards the first position with some series of allowed
alterations of the path.

The first question is if $v_1$ is in $X$. $X$ is a successful
pressing path of $G$ and $v_1$ is a black vertex in $G$ (since it is
the first vertex in the valid pressing path $Y$). Then either $v_1$
is pressed or it become a separated white vertex by pressing a
neighbor of $v_1$. Since $G$ is a linear graph, the only possibility
for the later case is that the remaining linear part of $G$ contains
two vertices: $v_1$ and some $u$, both of them are black and
connected, and $u$ is pressed in the pressing path. But then pressing
$v_1$ instead of $u$ has the same effect. Replacing $u$ to $v_1$ in
the pressing path means that the length of the longest common
subsequence is one less than the common length of the paths.

{\it Case 1.} So from now we assume that $v_1$ is part of the current pressing path,
which we denote by $P_1 w_1 v_1 P_2$, both $P_1$ and $P_2$ might be
empty. If $w_1$ and $v_1$ are not neighbors in $GP_1$, then $P_1 v_1
w_1 P_2$ is also a valid pressing path, and one of the longest common 
subsequences of $P_1 w_1 v_1 P_2$ and $P_1 v_1 w_1 P_2$ is $P_1 w_1
P_2$, one vertex less then the original pressing paths. In this way, we
can move $v_1$ to a smaller index position in the pressing path, and
this is what we want to prove.

{\it Case 2.} If $w_1$ and $v_1$ are neighbors, then $v_1$ is white in $G P_1$,
and then $w_1$ makes it black again. However, $v_1$ is black in
$G$, since it is the first vertex in the valid pressing path $Y$. Then
there have to be at least one vertex in $P_1$ that made $v_1$
white. Let $w_2$ be the last such vertex in $P_1$, and let we denote
$P_1 = P_{1a} w_2 P_{1b}$. We claim that none of the vertices in
$P_{1b}$ are neighbors of $w_2$ in $G P_{1a}$. Indeed, if there were a
neighbor of $w_2$ in $P_{1b}$, denote it by $w_3$, then $w_3$ would
become a neighbor of $v_1$ after pressing $w_2$, and then pressing
$w_3$ would make $v_1$ black, and then either $v_1$ was black before
pressing $w_1$, a contradiction, or there were further vertices in
$P_{1b}$ making $v$ white, contradicting that $w_2$ is the last such
vertex. Since $P_{1b}$ does not contain a vertex which is a neighbor
of $w_2$ in $G P_{1a}$, we can recursively bubble down $w_2$ next to
$w_1$. We get that the pressing path is now $P_1 w_2 w_1 v_1 P_2$,
where $P_1$ is now a different pressing path, and possibly empty, and
$P_2$ might also be empty. The topology and the colors of $w_2$, $w_1$
and $v_1$ in $G P_1$ is one of the following:

\begin{center}
\includegraphics[angle=0, width=3in]{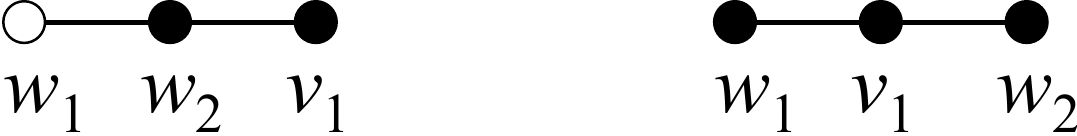}
\end{center}

{\it Case 2a}. Assume that $P_{2}$ is not empty, then the $\{w_1,
w_2, v_1\}$ triplet has at least one neighbor, call it $u$, and $u$
either is pressed in $P_{2}$, or we can replace a vertex in $P_{2}$
with $u$ such that it is still a successful pressing path on $G P_{1}
w_2 w_1 v_1$. So we can assume that at least one neighbor of the $\{w_1,
w_2, v_1\}$ triplet is pressed in $P_{2}$. It is easy to see that the
neighbors of the $\{w_1, w_2, v_1\}$ triplet changes their color in the
same way by pressing only $v_1$ and pressing $w_2 w_1 v_1$, see
Figure~\ref{fig:alter}. Therefore we can press $v_1$ instead of $w_2
w_1 v_1$, and the pressing path $P_{2}$ will be still valid up to the
point when $u_1$ or $u_2$ is pressed. Assume that $u_1$ is pressed
before $u_2$ in $P_{2}$, and $P_{2} = P_{2a} u_1 P_{2b}$
Figure~\ref{fig:finishing2a} shows
that the color of $u_2$  and a possible second neighbor of $u_1$
denoted by $u_3$ will be the same in $G P_1 w_2 w_1 v_1 P_{2a}
u_1$ and $G P_1 v_1 P_{2a} u_1 w_1 w_2$. Therefore $P_1 v_1 P_{2a} u_1 w_1
w_2 P_{2b}$ will be also a successful pressing path on $G$, since no
more vertices are affected by the given alteration of the pressing
path. One of the longest common subsequences of $ P_1 w_2 w_1 v_1
P_{2a} u_1 P_{2b}$ and $P_1 v_1 P_{2a} u_1 w_1 w_2 P_{2b}$ is $P_1 v_1
P_{2a} u_1 P_{2b}$, $2$ vertices less than the entire pressing
paths. $v_1$ is in a smaller index position of the pressing path, and
this is what we wanted to prove. The case when $u_2$ is pressed first
in $P_2$ is similar to the discussed case.

\begin{figure}
\center{
\includegraphics[angle=0, width=5in]{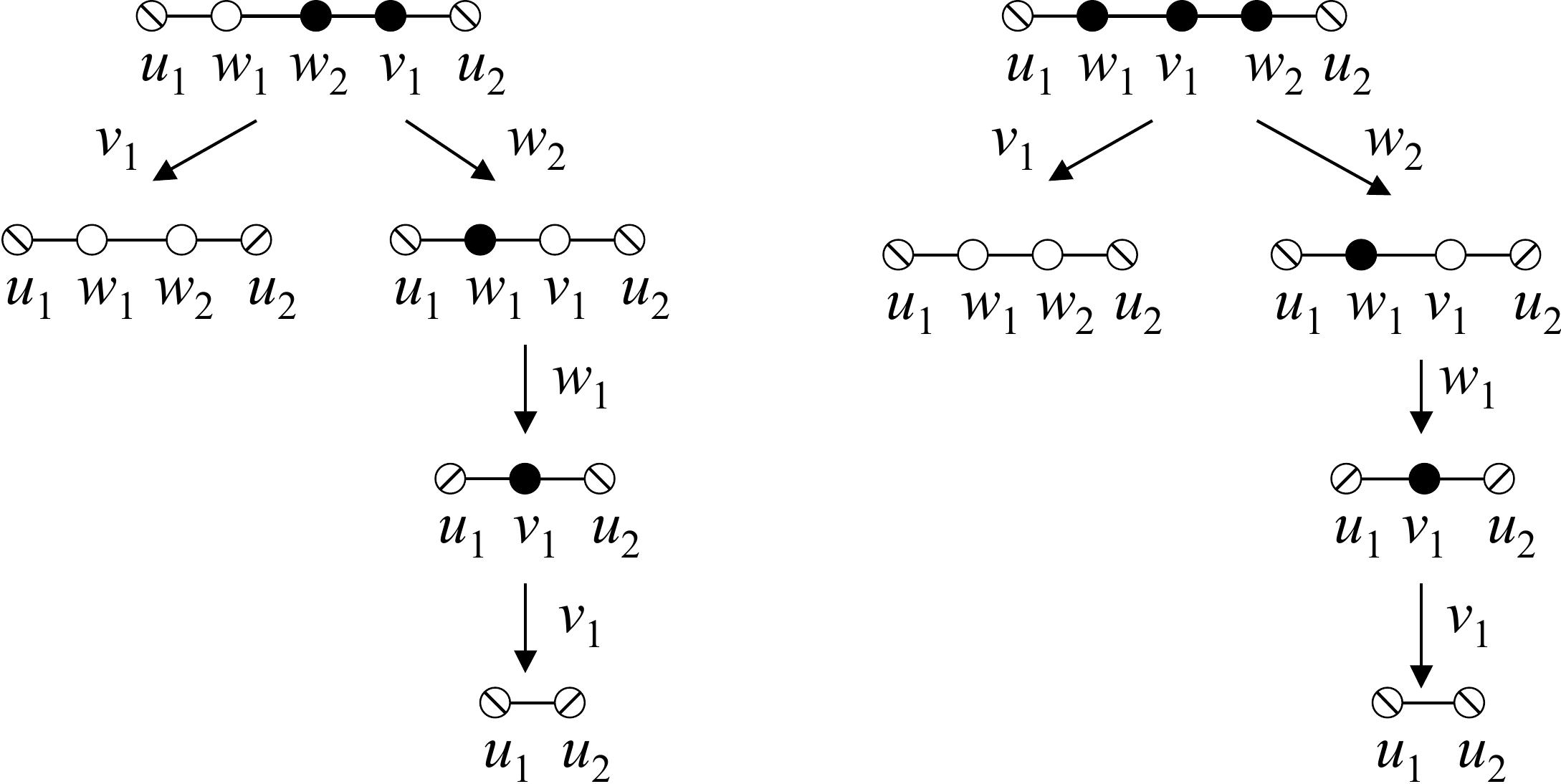}}
\caption{On the indicated two configurations, the neighbors of the
  $w_1,w_2,v_1$ triplet, $u_1$ and $u_2$ changes color in the same way
by pressing only $v_1$ and pressing $w_2 w_1 v_1$. The color change on
$u_1$ and $u_2$ is indicated with the flipping of their crossing line.}\label{fig:alter}
\end{figure}

\begin{figure}
\center{
\includegraphics[angle=0, width=3in]{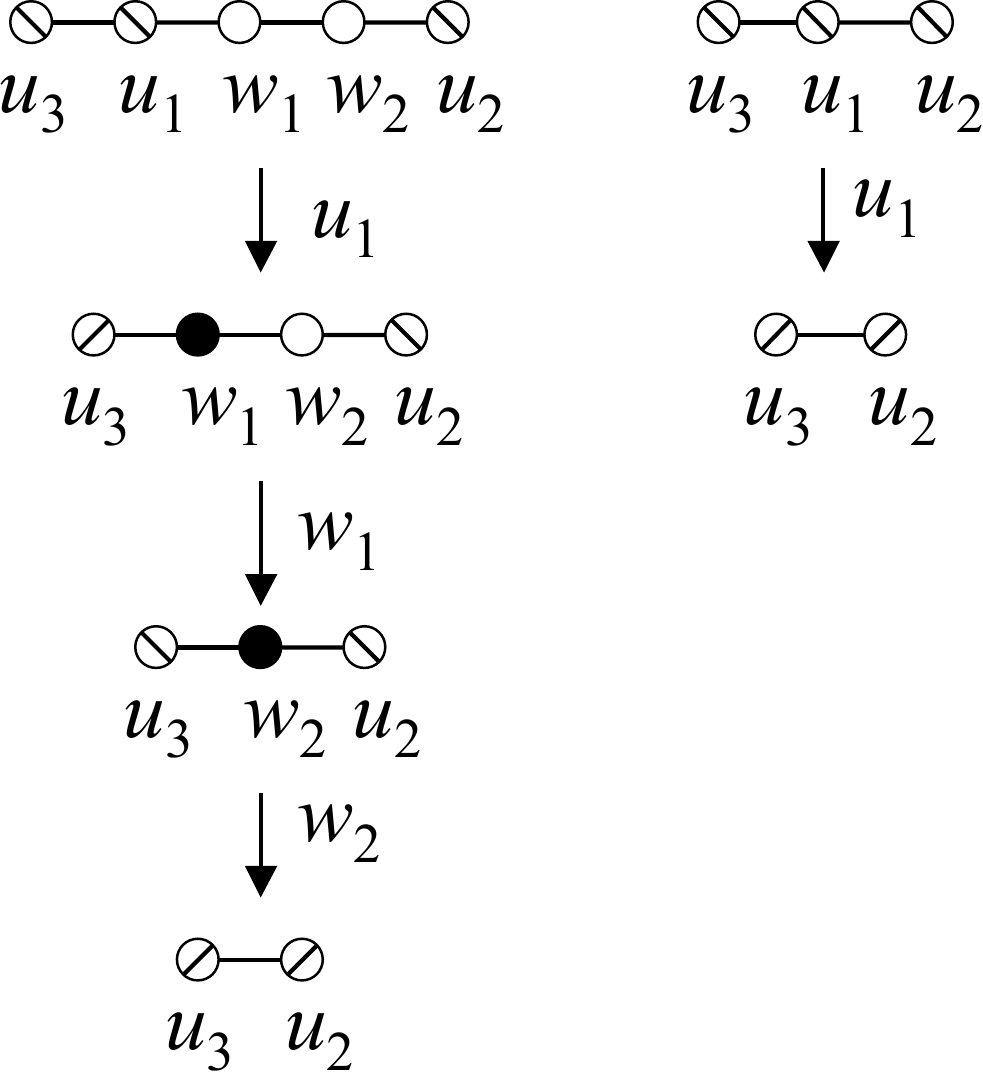}}
\caption{The color of $u_2$ and $u_3$ changes in the same way on the
  two indicated configurations. See text for details.}\label{fig:finishing2a}
\end{figure}

{\it Case 2b} Finally, assume that $P_2$ is empty. This means that
the $w_1, w_2, v_1$ triplet might have at most one more vertex that
becomes a separated white vertex when $v_1$ is pressed. This
additional vertex is white if a neighbor of $w_1$ or $w_2$ and black
if it is a neighbor of $v_1$ (it can be only when $w_2$ is a neighbor
of $w_1$.

Then $P_1$ cannot be empty, otherwise $w_2 w_1 v_1$ would be the only
successful pressing path, contradicting that a successful pressing
path exists that starts with $v_1$. 

If the last vertex in $P_1$ is a neighbor of $v_1$ when it is pressed,
then it makes $v_1$ black, namely, before pressing the last vertex in
$P_1$, $v_1$ is white. However, $v_1$ is black in $G$, so there has to
be further vertices in $P_1$ changing the color of $v_1$. The last
vertex in $P_1$ making $v_1$ white can be bubbled down to the last but
one position of $P_1$ just as we did with $w_2$. Let $P'$ be the path 
obtained from path $P_1$ in this way, excluding the last two
vertices. Then the graph $G P'$ contains the black vertex $v_1$, all
of its neighbors are black, and all further vertices are white. In
this graph, $v_1$ cannot be the first vertex of a successful pressing
path, since pressing it would create an all-white non-trivial
component. Then further vertices must be in $P'$. If the last vertex
of $P'$ is a neighbor of $v_1$, we can do the same thing, creating a
path $P''$ such that $G P''$ contains the black vertex $v_1$, all of
its neighbors are black, and all other vertices are white.

Since there is a successful pressing path which starts with $v_1$
after separating down a few -- possibly $0$ -- couples of vertices
from $P_1$, we have to find a vertex, call it $u$, which is not a
neighbor of $v_1$. Let the so-emerging pressing path be $P_{1a} u
P_{1b} v_1$. Note that we also incorporate $w_1$ and $w_2$ into
$P_{1b}$. The vertices in $P_{1b}$ are all neighbors of $v_1$ when
pressed, and at least one of them are neighbor of $u$. Let the left
neighbors of $v_1$ be denoted by $x_1, x_2 \ldots x_k$ and the let the
right neighbors be denoted by $y_1, y_2, \ldots y_l$. Without loss of
generality we can assume that $u$ is in the left neighbors (swap left
and right if this was not the case). Obviously, any $x$ is not a
neighbor of $y$, so we can rearrange them in $P_{1b}$ such that first
the $y$ vertices are pressed then the $x$ vertices. After a finite
number of allowed alterations, $P_{1b} = y_1 y_2 \ldots y_l x_1 x_2
\ldots x_k$ and $G P_{1a}$ is 

\begin{center}
\includegraphics[angle=0, width=4in]{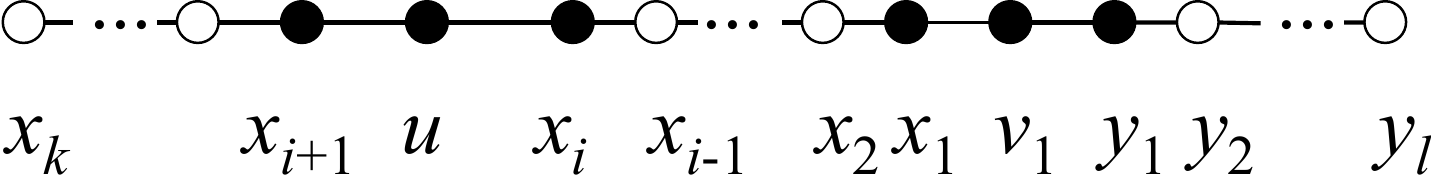}
\end{center}

Similarly, we move down vertex $u$ before $x_i$ in the pressing
path. We consider the graph $G P_{1a} y_1 \ldots y_l x_1 \ldots
x_{i-1}$ if $v_1$ is black in it (the runs of $x$ vertices might be
empty if $i=1$), and otherwise the graph $G P_{1a} y_1 \ldots y_l x_1
\ldots x_{i-2}$ (also the runs of $x$ vertices might be empty if
$i=2$) or $G P_{1a} y_1 \ldots y_{l-1}$ if $i = 1$. We have one of the
following graphs

\begin{center}
\includegraphics[angle=0, width=2.5in]{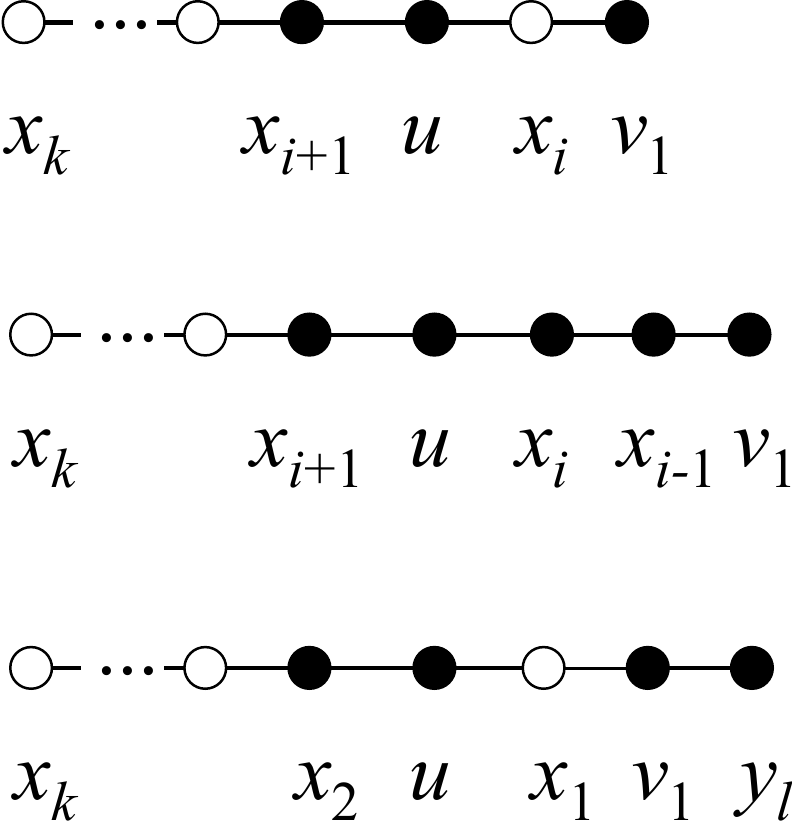}
\end{center}

on which $u x_i \ldots x_k v_1$, $u x_{i-1} \ldots x_k v_1$, $y_l u
x_1 \ldots x_k v_1$ is the current successful pressing path,
respectively.

A successful pressing path replacing $u x_i \ldots x_k v_1$ is
$v_1 x_i \ldots x_k u$, as can be seen on the left hand side of
Figure~\ref{fig:alt2a}. The length of the longest common subsequence
of the two pressing paths is $2$ less than their common length, as required.
\begin{figure}
\center{
\includegraphics[angle=0, width=4in]{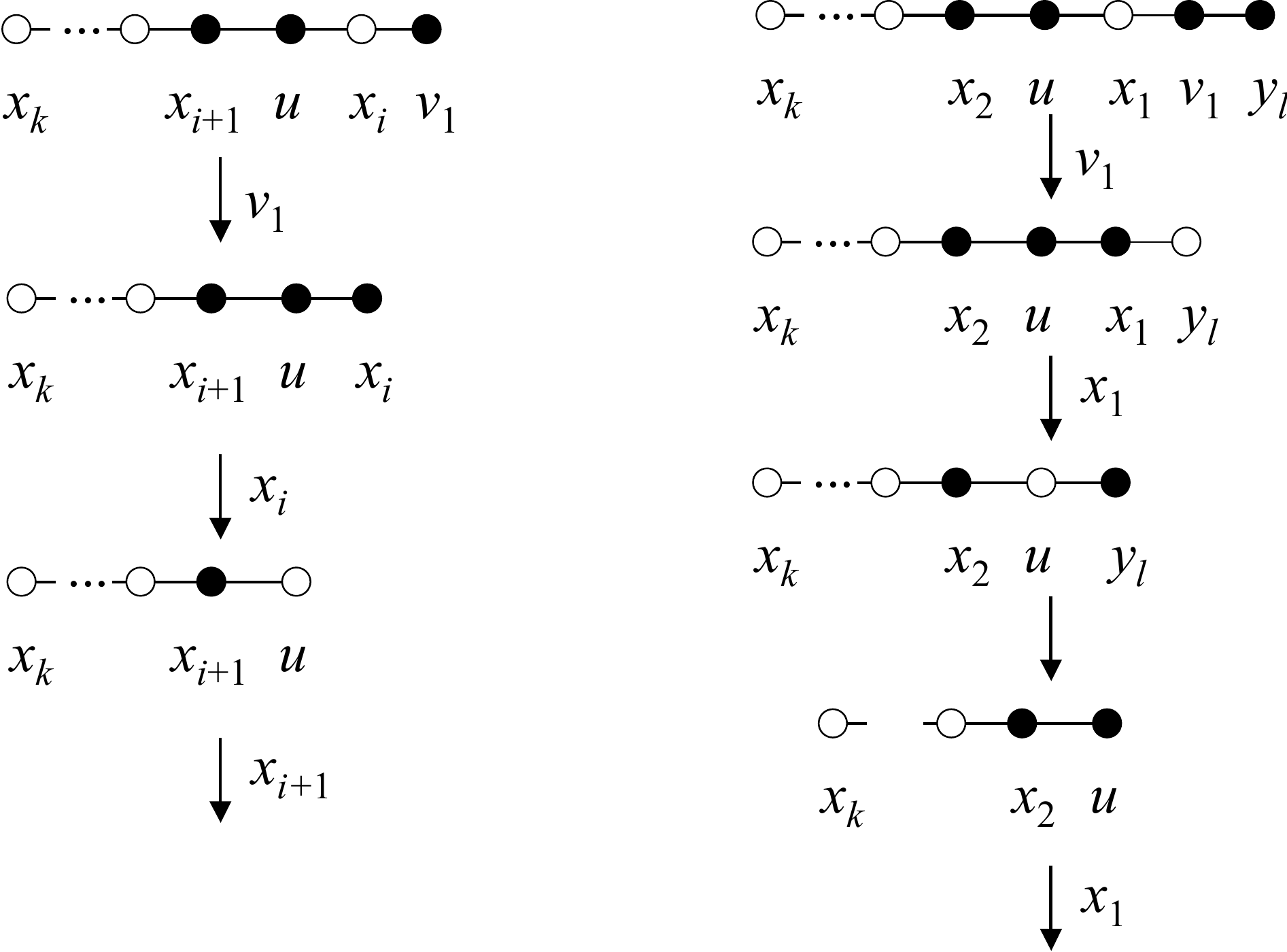}}
\caption{Alternative pressing paths for two cases. See text for details.}\label{fig:alt2a}
\end{figure}
The pressing path $y_l u x_1 \ldots x_k v_1$ can be replaced to $u x_1
y_l x_2 \ldots x_k v_1$ since $y_l$ is a neighbor neither $u$ nor
$x_1$. Then this pressing path can be replaced to $v_1 x_1 y_l x_2
\ldots x_k u$, as can be seen on the right hand side of
Figure~\ref{fig:alt2a}. The length of the longest common subsequence
of $u x_1 y_l x_2 \ldots x_k v_1$ and $v_1 x_1 y_l x_2 \ldots x_k u$
is again $2$ less than their common length.

Finally, the pressing path $u x_{i-1} \ldots x_k v_1$ can be replaced
in two steps, first it is changed to $x_i x_{i+1} u x_{i-1}
x_{i+2}\ldots x_k v_1$, then to $x_i x_{i+1} v_1 x_{i-1} x_{i+2}\ldots
x_k u$, as can be checked on Figure~\ref{fig:alt2b}. In both setps,
the length of the longest common subsequences of two consecutive
pressing paths is $2$ less than their common length as required.
\begin{figure}
\center{
\includegraphics[angle=0, width=4in]{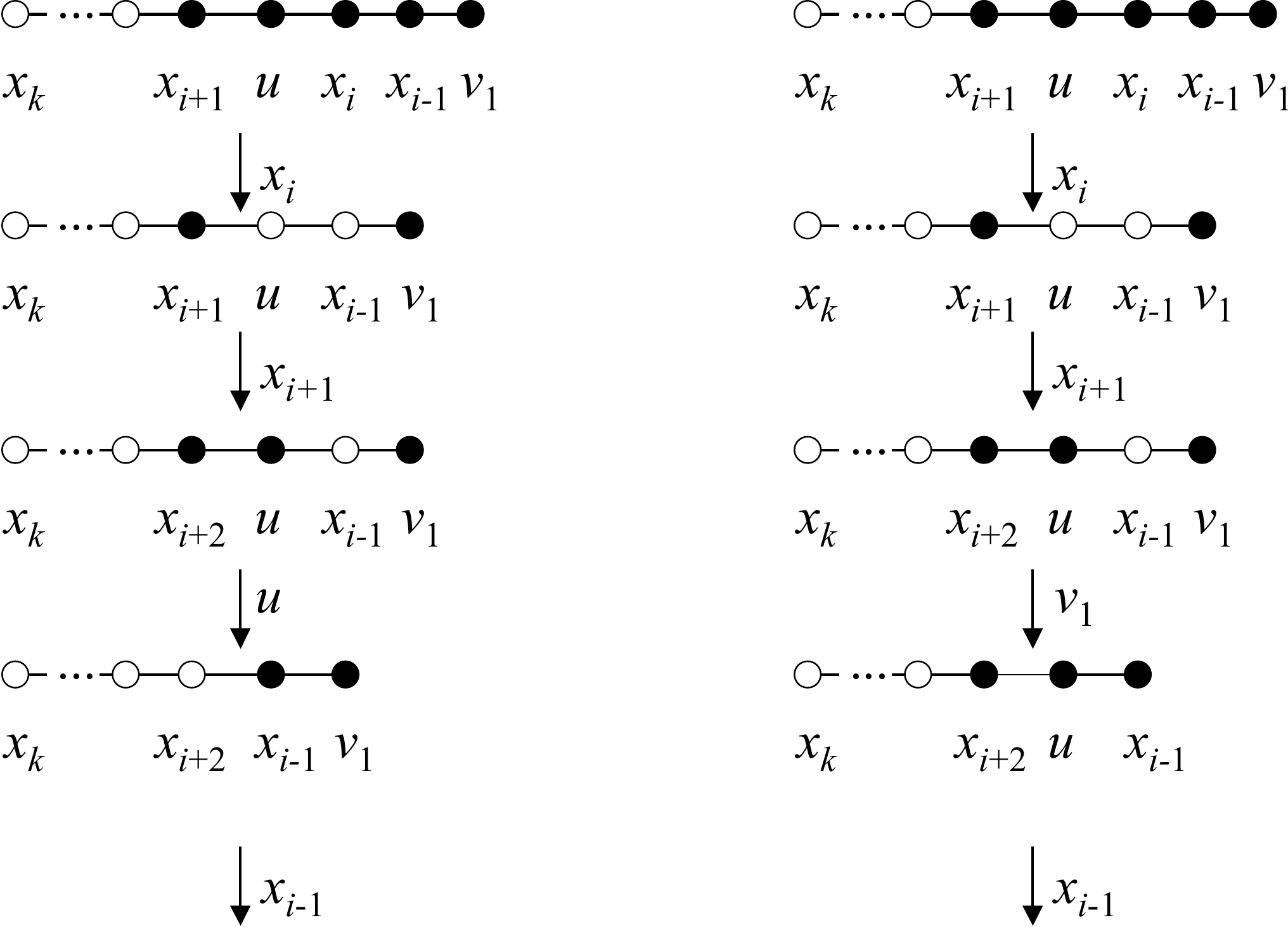}}
\caption{Changing the pressing path $u x_{i-1} \ldots x_k v_1$ in two
  steps such that $v_1$ is in a smaller index position. See text for details.}\label{fig:alt2b}
\end{figure}

We proved that in any case, $v_1$ can be moved into a smaller index
position with a finite series of allowed perturbations. Iterating
this, we can move $v_1$ to the first position. Then we can do the same
thing with $v_2$ on the graph $G v_1 \setminus \{v_1\}$, and eventually transform $X$
into $Y$ with allowed perturbations.

\end{proof}

\section{Discussion and Conslusions}\label{sec:conclusions}

In this paper, we proved the pressing game conjecture for linear
graphs. Although the linear graphs are very simple, the proving
technique shows a direction how to prove the general case. Indeed, it
is generally true that if a vertex $v$ is not in a successful pressing
path $P$, then a successful pressing path $P'$ exists which contains
$v$ and the length of the longest common subsequence of $P$ and $P'$
is only $1$ less than their common length. Case 1 in the proof of
Theorem~\ref{theo:main} holds for arbitrary graphs, and in a working
manuscript, we were able to prove that the conjecture is true for Case
2a using a linear algebraic techniques similar to that one used in
\cite{hv2006}. The only missing part is Case 2b, which seems to be
very complicated for general graphs.

A stronger theorem holds for the linear case that is conjectured for
the general case. One possible direction above proving the general
conjecture is to study the emerging Markov chain on the solution space
of the pressing game on linear graphs. We proved that a Markov chain
that randomly removes two vertices from the current pressing path,
adds two random vertices to it, and accepts it if the result is a
successful pressing path is irreducible. It is easy to set the jumping
probabilities of the Markov chain such that it converges to the
uniform distribution of the solutions. The remaing question is the
speed of convergence of this Markov chain. 

\section*{Acknowledgements}

I.M. was supported by OTKA grant PD84297. Alexey Medvedev is thanked
for fruitful discussions.

\end{document}